\documentclass[11pt,
fleqn]{article}
\usepackage{lscape,amssymb, amsmath,verbatim,graphicx}
\usepackage{epsfig,subfigure,float}
\usepackage{morefloats}
\usepackage{enumerate}
\usepackage{booktabs}
\usepackage{ragged2e}
\usepackage{footnote}

\newcommand{\sign}{\text{sign}}

\usepackage    {amsmath}
\usepackage    {amssymb}
\usepackage    {epsf}
\usepackage    {natbib}
\bibpunct{(}{)}{,}{a}{}{;}
\usepackage    {amsthm}

\newfont{\popis}{cmcsc10}

 \newtheorem{thm}{Theorem}
 \newtheorem{lem}{Lemma}

\theoremstyle{remark}

\theoremstyle{definition}

 \newcommand{\carka}{\raise0.2em\hbox{,}}

\newcommand{\begeqO}{\begin{eqnarray*}}
\newcommand{\eneqO}{\end{eqnarray*}}
\newcommand{\begeq}{\begin{eqnarray}}
\newcommand{\eneq}{\end{eqnarray}}

\newcommand{\nin}{\noindent}

\newcommand{\vph}{\varphi}
\newcommand{\lam}{\lambda}

\newcommand{\toprob}{\mathop{\to}\limits^{P}}

\newcommand{\tocalD}{\mathop{\longrightarrow}\limits^{\mathcal{D}}}

\newcommand{\bxi}{\boldsymbol \xi}

\newcommand{\bD}{\boldsymbol D}

\newcommand{\Sum}{\mathop{\sum}\limits}

\newcommand{\Int}{\mathop{\int}\limits}

\newcommand{\Max}{\mathop{\max}\limits}

\newcommand{\what}{\widehat}

\usepackage[T1]{fontenc}
\newcommand{\changefont}[3]{
\fontfamily{#1} \fontseries{#2} \fontshape{#3} \selectfont}

\begin{document}
\renewcommand{\thefootnote}{\fnsymbol{footnote}}

\begin{savenotes}
\parbox[t]{400pt}{
\changefont{phv}{b}{n}
\nin {
\LARGE Testing the Equality of Covariance Operators in Functional Samples\footnote{Research partially
supported by NSF grants DMS 0905400 at the
University of Utah, DMS-0804165 and DMS-0931948 at Utah State University and DFG grant STE 306/22-1 at the University of Cologne.}
}
\renewcommand{\thefootnote}{\arabic{footnote}}
}
\end{savenotes}
\bigskip\\
STEFAN FREMDT\\
{\small\textit{ Mathematical Institute, University of Cologne}}\bigskip\\
LAJOS HORV\'{A}TH\\
{\small \textit{Department of Mathematics, University of Utah}}\bigskip\\
PIOTR KOKOSZKA\\
{\small \textit{Department of Mathematics and Statistics,
Utah State University}}\bigskip\\
JOSEF G. STEINEBACH\\
{\small\textit{ Mathematical Institute, University of Cologne}}\\
\parbox[t]{420pt}{
\paragraph{ABSTRACT.}{
\bf\small
We propose a robust test for the equality of the covariance
structures in two functional samples. The test statistic has a chi-square
asymptotic distribution with a known number of degrees of freedom,
which depends on the level of dimension reduction needed to represent
the data. Detailed analysis of the asymptotic properties is
developed. Finite sample performance is examined by a simulation
study and an application to egg--laying curves of fruit flies.}

\vspace{2mm}

\noindent {\em Abbreviated Title:} Equality of covariance operators

\vspace{2mm}

}

\allowdisplaybreaks

\bigskip

\section{Introduction}\label{sec0}

The last decade has seen increasing interest in methods of functional
data analysis which offer novel and effective tools for dealing with
problems where curves can naturally be viewed as data objects. The
books by \citet{ramsay:silverman:2005} and
\citet{ramsay:hooker:graves:2009} offer comprehensive introductions to
the subject, the collection \citet{ferraty:romain:2011} reviews some
recent developments focusing on advances in the relevant theory, while
the monographs of \citet{bosq:2000}, \citet{ferraty:vieu:2006} and
\citet{HKbook} develop the field in several important directions.
Despite the emergence of many alternative ways of looking at
functional data, and many dimension reduction approaches, the
functional principal components (FPC's) still remain the most
important starting point for many functional data analysis procedures,
\citet{reiss:ogden:2007}, \citet{gervini:2008}, \citet{yao:muller:2010},
\citet{gabrys:horvath:kokoszka:2010} are just a handful of illustrative
references.  The FPC's are the eigenfunctions of the covariance
operator. This paper focuses on testing if the covariance operators of
two functional samples are equal. By the Karhunen-Lo\`eve expansion,
this is equivalent to testing if both samples have the same set of FPC's. \citet{benko:hardle:kneip:2009} developed bootstrap procedures
for testing the equality of specific FPC's.  \citet{panaretos:2010}
proposed a test of the type we consider, but assuming that the curves
have a Gaussian distribution. The main
result of \citet{panaretos:2010} follows as a corollary of our more
general approach (Theorem~\ref{thm 3.1}).

Despite their importance, two sample problems for functional data
received relatively little attention.  In addition to the work of
\citet{benko:hardle:kneip:2009} and \citet{panaretos:2010}, the relevant
references are \citet{horvath:kokoszka:reimherr:2009} and
\citet{horvath:kokoszka:reeder:2011} who focus, respectively, on the
regression kernels in functional linear models and the mean of
functional data exhibiting temporal dependence.  Clearly, if some
population parameters of two functional samples are different,
estimating them using the pooled sample may lead to spurious
conclusions. Due to the importance of the FPC's, a relatively simple
and robust procedure for testing the equality of the covariance
operators is called for.

The remainder of this paper is organized as follows. Section
\ref{s:pre} sets out the notation and definitions. The construction
of the test statistic and its asymptotic properties are
developed in Section~\ref{sec1}.  Section~\ref{s:sim} reports
the results of a simulation  study and illustrates the procedure by
application to egg-laying curves of Mediterranean fruit flies.
The proofs of the asymptotic results of  Section~\ref{sec1}
are given in Section~\ref{sec2}.

\section{Preliminaries} \label{s:pre}
Let $X_1, X_2, \ldots, X_N$ be independent, identically distributed
random variables in $L_2 [0,\hspace{-1.5pt}1]$ with $EX_i(t)=\mu(t)$ and $cov (X_i
(t), X_i(s))= C(t,s)$. We assume that another sample $X^\ast_1,
X^\ast_2,\ldots X^\ast_M$ is also available and let $\mu^\ast (t)=E
X^\ast_i (t)$ and $C^\ast(t,s)=cov(X^\ast_i (t), X^\ast_i(s))$ for
$t,s \in [0,1].$ We wish to test the null hypothesis
\[
H_0: ~C=C^\ast
\]
against the alternative $H_A$ that $H_0$ does not hold.

A crucial assumption considering the asymptotics of our test procedure will be that
\begin{align}\label{eq 1.4}
\Theta_{N,M} =\frac {N}{M+N} \to \Theta \in (0,1)\quad\text{as}\quad N,M\to \infty.
\end{align}
For the construction of our test procedure we will use an estimate of the asymptotic pooled covariance operator $\mathfrak{R}$ of the two given samples (cf. \eqref{eq 1.6}) which is defined by the kernel
\[
 R(t,s) = \Theta C(t,s) + (1-\Theta)C^\ast(t,s).
\]
Denote by
$(\lambda_1, \varphi_1), (\lambda_2,\varphi_2),\ldots$
the eigenvalue/eigenfunction pairs of $\mathfrak{R}$, which are defined by
\begeq\label{eq 1.1}
\lambda_k \varphi_k (t)=\mathfrak{R}\varphi_k(t)
=\int^1_0 R(t,s)\varphi_k (s) ds,\quad t\in[0,1],\quad 1\le k<\infty~.
\eneq
Throughout this paper we assume
\begeq\label{eq 1.2}
\lambda_1>\lambda_2>\ldots >\lambda_p>\lambda_{p+1},
\eneq
i.e. there exist at least $p$ distinct (positive) eigenvalues.
Under assumption (\ref{eq 1.2}), we can uniquely (up to signs)
choose $\varphi_1,\ldots,\varphi_p$
satisfying (\ref{eq 1.1}), if we require $\|\varphi_i\|=1$, where for a positive integer $d$ and for $x\in L_2\left([0,1]^d\right)$
\[
\|x \|=\left( \int^1_0\cdots\int^1_0 x^2 (t_1,\ldots,t_d) dt_1\cdots dt_d \right)^{1/2}.
\]

Thus, under (\ref{eq 1.2}),  $\{\varphi_i, 1\le i \leq p\}$ 
is an orthonormal system that can be extended to an orthonormal basis $\{\varphi_i, 1\le i <\infty\}$.

If $H_0$ holds, then $(\lambda_i, \varphi_i), ~1\le i<\infty,$ are also
the eigenvalues/eigenfunctions of the covariance operators
$\mathfrak{C}$ of the first and $\mathfrak{C}^\ast$ of the
second sample.  To construct a test statistic which converges
under $H_0$, we can therefore pool the two samples, as explained
in Section~\ref{sec1}.

\section{The test and the
asymptotic results}
\label{sec1}
Our procedure is
based on  projecting the observations onto a suitably chosen
finite-dimensional space. To define this space,
introduce the empirical pooled covariance operator
$\what{\mathfrak{R}}_{N,M}$ defined by  the kernel
\begin{align}\label{eq 1.6}
\what{R}_{N,M}(t,s)=\frac 1{N+M} \Biggl\{&\Sum^N_{k=1}(X_k(t)-
\overline X_N (t))(X_k(s)-\overline X_N (s))
\\
&+\Sum^M_{k=1}(X^\ast_k(t)-\overline X_M^* (t))(X^\ast_k(s)-
\overline X_M^* (s))\Biggr\},\notag
\end{align}
where
\[\overline X_N(t)=\frac 1N\Sum^N_{k=1} X_k (t)\quad\hbox{and}\quad\overline X^\ast_M(t)
=\frac 1M \sum^M_{k=1} X^\ast_k (t)
\]
are the sample mean functions. Let $(\what\lambda_i, \what\varphi_i)$
denote the eigenvalues/eigenfunctions of $\what{\mathfrak{R}}_{N,M}$, i.e.
\[
\what\lam_i\what\vph_i (t)=\what{\mathfrak{R}}_{N,M}\what\vph_i(t)
=\Int^1_0 \what{R}_{N,M}(t,s)\what\vph_i(s)ds,\quad t\in[0,1],\quad 1\le i\le N+M,
\]
with
$
\what\lam_1 \ge \what\lam_2\ge\ldots~.
$
We can and will assume that the $\what\vph_i$ form an orthonormal system.
We consider the  projections
\begeq\label{eq 3.2}
\what a_k(i)= \langle X_k-\overline X_N,\what\vph_i\rangle
=\Int^1_0(X_k(t)-\overline X_N(t))\what\vph_i (t) dt
\eneq
and
\begeq\label{eq 3.3}
\what a^\ast_k(j)= \langle X^\ast_k-\overline X^\ast_M,\what\vph_j\rangle
=\Int^1_0\left(X^\ast_k(t)-\overline X^\ast_M(t)\right)\what\vph_j (t) dt.
\eneq
To test $ H_0$, we compare the matrices $\what\Delta_N$
and $\what\Delta^\ast_M$ with entries
\[
\what\Delta_N(i,j)
=\frac 1N\Sum^N_{k=1}\what a_k(i)\what a_k(j),\quad 1\le i,j\le p,
\]
and
\[
\what\Delta^\ast_M(i,j)=\frac 1M \Sum^M_{k=1}\what a^\ast_k(i)
\what a^\ast_k(j),\quad 1\le i,j\le p.
\]
We note that $\what\Delta_N(i,j)-\what\Delta^\ast_M(i,j)$ is the projection of
$\what{C}_N(t,s) -\what{C}^\ast_M(t,s)$ in the direction of
$\what\vph_i(t)\what\vph_j(s)$, where
\[
\what C_N(t,s)=\frac 1N \Sum^N_{k=1}\left(X_k(t)-\overline X_N(t)\right)
\left(X_k(s)-\overline X_N(s)\right)
\]
and
\[\what C^\ast_M(t,s)=\frac 1M \Sum^M_{k=1}
\left(X^\ast_k(t)-\overline X^\ast_M(t)\right)
\left(X^\ast_k(s)-\overline X^\ast_M(s)\right)
\]
are the empirical covariances of the two samples.

From the columns below the diagonal of $\what\Delta_N-\what\Delta^\ast_M$
we create a vector $\what{\bxi}_{N,M}$ as follows:
\begeq
\what{\bxi}_{N,M}= \mbox{vech} \left(\what\Delta_N-\what\Delta^\ast_M \right) = \left(
\begin{array}{c}
\what\Delta_N(1,1)-\what\Delta^\ast_M(1,1)
\\
\what\Delta_N(2,1)-\what\Delta^\ast_M(2,1)
\\
\vdots
\\
\what\Delta_N(p,p)-\what\Delta^\ast_M(p,p)
\end{array}
\right).
\eneq
For the properties of the vech operator we refer to \citet{abadir:magnus:2005}.\\
Next we estimate the asymptotic covariance matrix of
$(MN/(N+M))^{1/2}\what{\bxi}_{N,M}.$ Let
\begin{align*}
\what L_{N,M} (k,k')= &\left(1-\Theta_{N,M}\right)\left\{ \frac 1N \Sum^N_{\ell=1}\what a_\ell(i)
\what a_\ell(j)\what a_\ell(i')\what a_\ell(j')
-\big\langle \what{\mathfrak C}_N \what\vph_i,\what\vph_j\big\rangle
\big\langle\what{\mathfrak C}_N\what\vph_{i'},\what\vph_{j'}\big\rangle\right\}
\\
&+ \Theta_{N,M}\left\{ \frac 1M \Sum^M_{\ell=1}
\what a^\ast_\ell(i)\what a^\ast_\ell(j)\what a^\ast_\ell(i')
\what a^\ast_\ell(j')
-\big\langle \what{\mathfrak C}^\ast_M \what\vph_i,\what\vph_j\big\rangle
\big\langle\what{\mathfrak C}^\ast_M\what\vph_{i'},\what\vph_{j'}\big\rangle\right\}
\end{align*}
where $i,j,i',j'$ depend on $k,k'$ (see below),
and $\what{\mathfrak C}_N$ $(\what{\mathfrak C}^\ast_M)$ is interpreted as
an operator with $\what{\mathfrak C}_N$ defined as
\[
 \what{\mathfrak C}_N \what\vph_i =\int^1_0 \what C_N(t,s)\what\vph_i (s) ds.
\]
(An analogous definition holds for $\what{\mathfrak C}^\ast_M$.)
From this definition it follows that
\[
\big\langle \what{\mathfrak C}_N \what\vph_i,\what\vph_j\big\rangle
=\frac 1N \Sum^N_{l=1} \what a_\ell (i) \what a_\ell (j).
\]
We note that one can use $\what L^\ast_{N,M}(k,k')$ instead of
 $\what L_{N,M}(k, k'),$ where $\what L^\ast_{N,M}(k,k')$ is defined like
 $\what L_{N,M}(k,k'),$ but $\big\langle \what{\mathfrak C}_N \what\vph_i,\what\vph_j\big\rangle$ and
 $\big\langle \what{\mathfrak C}^\ast_M \what\vph_i,\what\vph_j\big\rangle$
 are replaced with $0$ if $i\neq j$ and $\what\lam_i$ if $i=j.$
 In the same spirit, $\big\langle \what{\mathfrak C}_N \what\vph_{i'},\what\vph_{j'}\big\rangle$ and
 $\big\langle \what{\mathfrak C}^\ast_M \what\vph_{i'},\what\vph_{j'}\big\rangle$ are replaced
 with $0$ for ${i'}\neq {j'}$ and $\what\lam_{i'}$ if ${i'}={j'}.$

The index $(i,j)$ is computed from $k$ in the following way: Let
\begin{align} \label{index1}
k'=\frac{p(p+1)}2 -k+1, \quad i'=p-i+1, \quad\mbox{and} \quad j'= p-j+1.
\end{align}
We look at an upper triangle matrix $(a_{i',j'})$. Then, for column $j',$
we have that $(j'-1)j'/2 < k\le j'(j'+1)/2.$
Thus $j'=\left\lceil\sqrt{2k'+\frac 14}-\frac 12\right\rceil$ and
$i'=k'-(j'-1)j'/2,$ where $\lceil r\rceil=\min\{k\in\mathbb Z:\,k\ge r\}$
for $r\in\mathbb{R}.$
Consequently, the index $(i,j)$ can be computed from $k$ via
\begin{align} \label{index2}
j=p-\left\lceil\sqrt{p(p+1)-2k+\frac 94}-\frac 12\right\rceil+1
\quad\mbox{and}\quad i=k+p-p\cdot j+\frac{j(j-1)}2.
\end{align}

With the above notation, we can formulate the main result of this paper:

\begin{thm}\label{thm 1.1}
We assume that $H_0$, (\ref{eq 1.4}) and (\ref{eq 1.2}) hold, and
\begin{align}\label{eq 1.3}
\int^1_0 E (X_1(t))^4dt<\infty, \quad \int^1_0 E(X^\ast_1(t))^4 dt <\infty.
\end{align}
Then
\[
\frac{NM}{N+M} \what{\bxi}^T_{N,M} \what L^{-1}_{N,M} \what{\bxi}_{N,M}~\tocalD ~\chi^2_{p(p+1)/2},\quad\text{as}\quad N,M\to \infty
,
\]
where $\chi^2_{p(p+1)/2}$ stands for a $\chi^2$ random variable with $p(p+1)/2$
degrees of freedom.
\end{thm}

Theorem~\ref{thm 1.1} implies that the null hypothesis is rejected
if the test statistic
\[
\what T_1
= \frac{NM}{N+M} \what{\bxi}^T_{N,M} \what L^{-1}_{N,M} \what{\bxi}_{N,M}
\]
exceeds a critical quantile of the chi--square distribution
with $p(p+1)/2$ degrees of freedom.

If both samples are Gaussian random processes, the quadratic form
$\what{\bxi}^T_{N,M} \what L_{N,M}^{-1}\what{\bxi}_{N,M}$
can be replaced with the normalized sum of the squares of
$\what\Delta_{N,M} (i,j) -\what\Delta^\ast_{N,M} (i,j)$, as stated
in the following theorem.

\begin{thm}\label{thm 3.1}
If $X_1, X^\ast_1$ are Gaussian processes and the conditions of
Theorem \ref{thm 1.1} are satisfied, then, as $N,M\to \infty$,
\[
\what T_2 =
\frac{NM}{N+M}\Sum_{1\le i,j\le p}\frac 12
\frac{\left(\what\Delta_{N,M} (i,j)-
\what\Delta^\ast_{N,M}(i,j)\right)^2}
{\what\lam_i \what\lam_j} ~\tocalD~ \chi^2_{p(p+1)/2}.
\]
\end{thm}

Observe that the statistic $\what T_2$ can be written as
\[
\what T_2 =
\frac{NM}{N+M}\left\{\Sum_{1\le i< j\le p}
\frac{\left(\what\Delta_{N,M} (i,j)-
\what\Delta^\ast_{N,M}(i,j)\right)^2}
{\what\lam_i \what\lam_j}
+ \Sum^p_{i=1}
\frac{\left(\what\Delta_{N,M}(i,i)-
\what\Delta^\ast_{N,M}(i,i)\right)^2}
{2 \what\lam_i^2}\right\}.
\]

Next we discuss the asymptotic consistency of the testing procedure based on Theorem \ref{thm 1.1}.
Analogously to the definition of $\what \bxi_{N,M}$ we define the vector $\bxi = (\xi(1),\ldots,\xi(p(p+1)/2))$ using the columns of the matrix
\begin{align}\label{eq 3.1}
\bD =
\left(\Int_0^1\Int_0^1(C(t,s) - C^\ast(t,s))\vph_i(t)\vph_j(s)dt\,ds\right)_{i,j = 1,\ldots,p}
\end{align}
instead of $\what\Delta_N-\what\Delta^\ast_M$, i.e.
\[
 \what \bxi = \mbox{vech}(\bD).
\]

\begin{thm}\label{thm 3.3}
We assume that $H_A$, (\ref{eq 1.4}), (\ref{eq 1.2}) and (\ref{eq 1.3}) hold.
Then there exist random variables $\what{h}_1 = \what{h}_1(N,M),\ldots,\what{h}_{p(p+1)/2} = \what{h}_{p(p+1)/2}(N,M)$, taking values in $\{-1,1\}$ such that, as $N,M\to \infty$,
\begin{align}\label{eq 3.5}
 \max_{1\leq i\leq p(p+1)/2}\left| \what{\xi}_{N,M}(i) - \what{h}_i\xi(i)\right| = o_P (1)
\end{align}
and therefore
\begin{align}\label{eq 3.6}
 \left|\what{\bxi}_{N,M}\right| \toprob \left|\bxi\right|,
\end{align}
where $|\cdot |$ denotes the Euclidean norm. If $\bxi \neq \mathbf{0}$ and the $p$ largest eigenvalues of $C$ and $C^\ast$ are positive, we also have
\begin{align}\label{eq 3.7}
 \what T_1 \toprob \infty,\quad\text{as}\quad N,M\to \infty.
\end{align}
\end{thm}
The assumption that the $p$ largest eigenvalues of $C$ and $C^\ast$ 
are positive implies that the random functions 
$X_i,~i = 1,\ldots,N$, and $X^\ast_j,~j = 1,\ldots,M$, are not included
in a $(p-1)$-dimensional subspace.\\

The application of the test requires the selection of the number $p$
of the empirical FPC's to be used. A rule of thumb is to choose $p$
so that the first $p$ empirical FPC's in each sample (i.e.
those calculated as the eigenfunctions of $\what C_N$ and  $\what C^\ast_M$)
explain about 85--90\% of the variance in each sample.
Choosing $p$ too large generally negatively affects the finite sample
performance of tests of this type, and for this reason we do not study
asymptotics as $p$ tends to infinity. It is often illustrative
to apply the test for a range of the values of $p$; each $p$ specifies
a level of relevance of differences  in the curves or kernels.
A good practical approach is to look at the Karhunen--Lo\`eve approximations
of the curves in both samples, and choose $p$ which gives approximation
errors that can be  considered unimportant.

\section{A simulation study
and application to medfly data}
\label{s:sim}
We first describe the results of a simulations study designed 
to evaluate finite sample properties of the tests based on the
statistics $\what T_1$ and $\what T_2$.
The emphasis is on the robustness to the violation 
of the assumption of normality. 
We simulated Gaussian curves
as Brownian motions and  Brownian bridges, and non--Gaussian curves via
\begin{equation} \label{e:Zn}
Z_n(t) = A_n \sin(\pi t) + B_n
\sin(2\pi t)+ C_n \sin(4\pi t),
\end{equation}
where
$A_n = 5Y_1$, $B_n = 3Y_2$, $C_n = Y_3$ and $Y_1,Y_2,Y_3$ are
independent $t_5$-distributed random variables. All curves were
simulated  at 1000 equidistant points in the interval $[0,1]$, and
transformed into functional data objects using  the Fourier basis with
49 basis functions. For each data generating process we used one thousand
replications.

\begin{table}
\caption{Empirical sizes of the tests based on statistics $\what T_1$
and $\what T_2$ for non--Gaussian data.
The curves in each sample were generated according to
(\ref{e:Zn}).
}
\begin{center}
\begin{tabular}{lcccccc} \toprule
&\multicolumn{6}{c}{$p =2$}\\
\cmidrule{2-7}
&\multicolumn{3}{c}{$\what T_1$}&\multicolumn{3}{c}{$\what T_2$}\\
\cmidrule(r){2-4}\cmidrule(l){5-7}
Sample Sizes&1\%&5\%&10\%&1\%&5\%&10\%\\\hline
$N=M=100$&0.005&0.028&0.061&0.152&0.275&0.380 \\
$N=M=200$&0.003&0.021&0.058&0.163&0.314&0.402\\
$N=M=1000$&0.002&0.021&0.056&0.190&0.313&0.426\\
\midrule
&\multicolumn{6}{c}{$p=3$}\\
\cmidrule{2-7}
&\multicolumn{3}{c}{$\what T_1$}&\multicolumn{3}{c}{$\what T_2$}\\
\cmidrule(r){2-4}\cmidrule(l){5-7}
Sample Sizes&1\%&5\%&10\%&1\%&5\%&10\%\\\hline
$N=M=100$&0.004&0.028&0.065&0.167&0.332&0.434\\
$N=M=200$&0.004&0.024&0.064&0.194&0.338&0.423\\
$N=M=1000$&0.004&0.028&0.070&0.240&0.384&0.484\\
\bottomrule
\end{tabular}
\end{center}
\label{tb1}
\end{table}

Table~\ref{tb1}  displays the empirical sizes for non--Gaussian data.  The
test based on $\what T_2$ has severely inflated size, due to the
violation of the assumption of normality. As documented in
\citet{panaretos:2010}, and confirmed by our own simulations, this test
has very good empirical size when the data are Gaussian.  The test
based on $\what T_1$ is conservative, especially for smaller sample
sizes. This is true for both Gaussian and non--Gaussian data; there is
not much difference in the empirical size of this test for different
data generating processes. Reflecting its conservative size, statistic
$\what T_1$ leads to smaller power than $\what T_2$.  We also studied
a Monte Carlo version of the test based on the statistic $\what T_3 =
{NM}(N+M)^{-1} \what{\bxi}^T_{N,M} \what{\bxi}_{N,M}$,  and found that its
finite sample properties were similar to those of the test based on
$\what T_1$.

\begin{figure}
\caption{Ten randomly selected smoothed egg-laying curves of short-lived
medflies (left panel), and ten such curves for long--lived medflies
(right panel).}
\begin{center}
\includegraphics[scale = 0.465,page = 1]{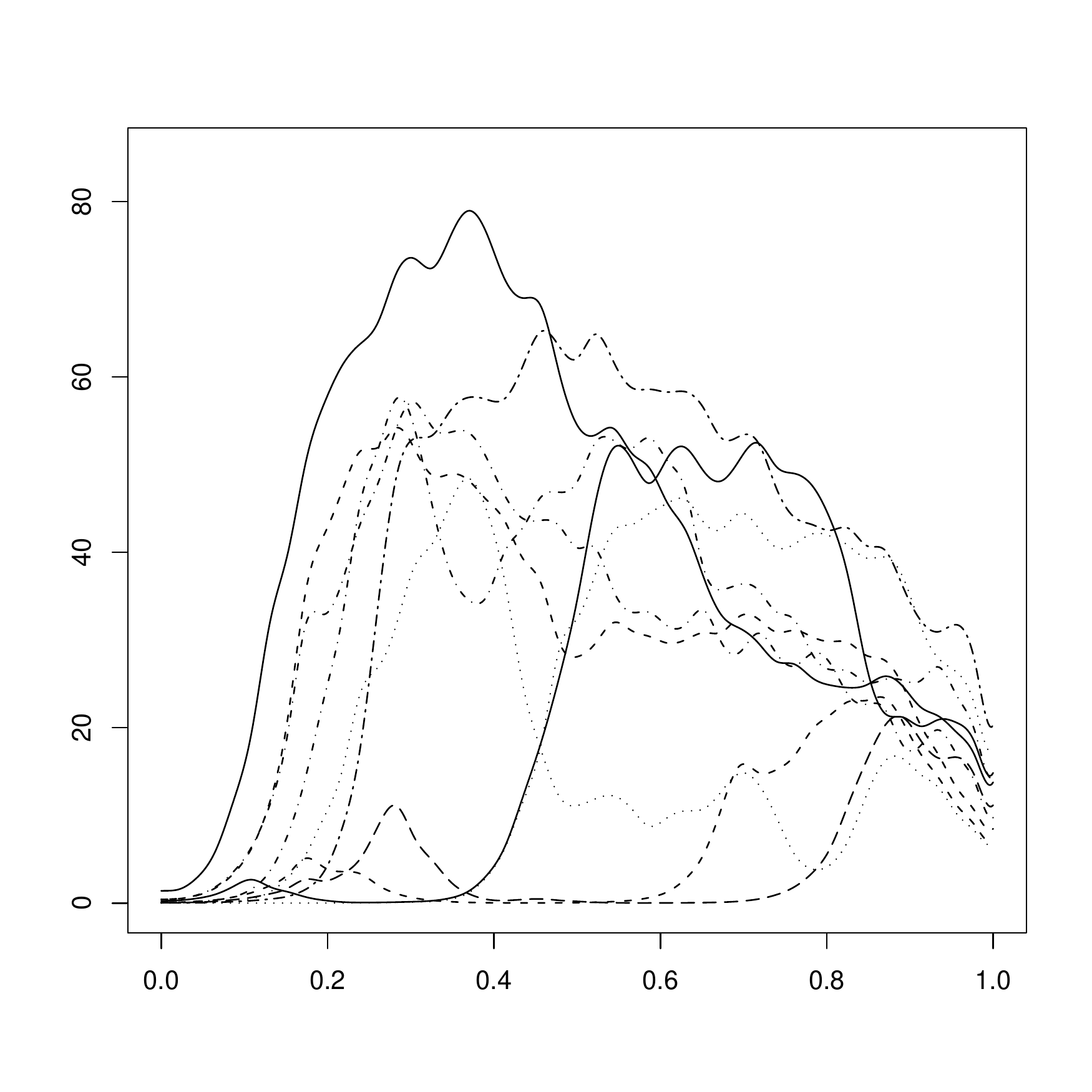}
\includegraphics[scale = 0.465,page = 2]{smoothed-functions-bspline-L49.pdf}
\end{center}
\label{f:egg-curves}
\end{figure}

\begin{figure}
\caption{Ten randomly selected smoothed egg-laying curves of short-lived
medflies (left panel), and ten such curves for long--lived medflies
(right panel), relative to the number of eggs laid in the fly's lifetime.}
\begin{center}
\includegraphics[scale = 0.465,page = 1]{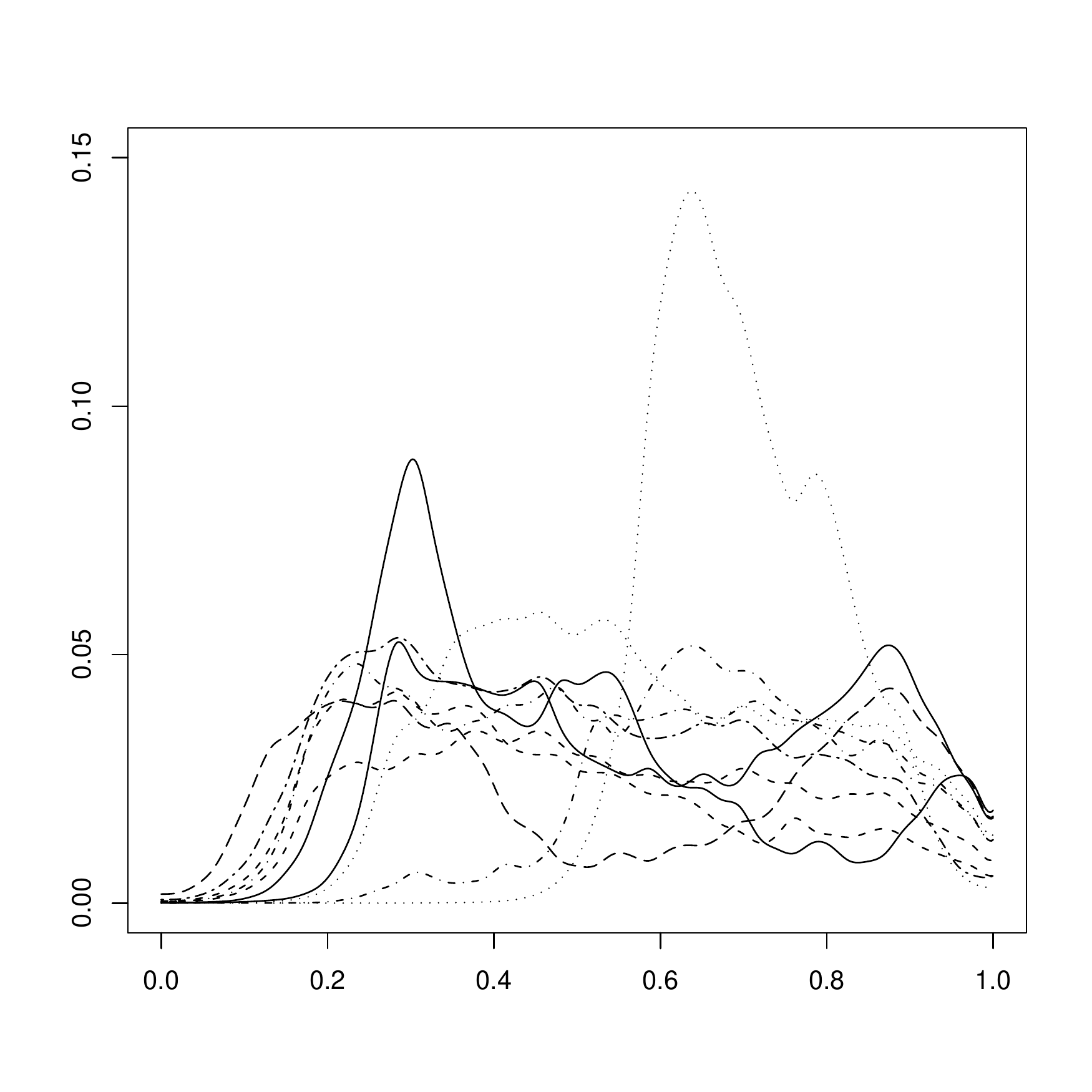}
\includegraphics[scale = 0.465,page = 2]{smoothed-rel-functions-bspline-L49.pdf}
\end{center}
\label{f:egg-curves-rel}
\end{figure}

We now describe the results of the application of both tests 
to an interesting data set
consisting of egg--laying trajectories of Mediterranean fruit flies
(medflies). The data  were kindly made available to
us by Hans--Georg M{\"u}ller. This data set
has been  extensively studied in biological
and statistical literature, see \citet{muller:stadtmuller:2005}
and references therein. We consider 534 egg-laying
curves 
of medflies who lived at least 34 days.
We examined two versions of these egg-laying curves,
the functions in either version are defined over an interval $[0,30]$,
and  $t\le 30$ is the day. 
Version 1 curves  (denoted $X_i(t)$) are 
the absolute counts of eggs laid by fly $i$ on 
day $t$. 
Version 2 curves (denoted $Y_i(t)$) 
are the counts of eggs laid by fly $i$ on  day $t$
{\em relative}  to the total number of eggs laid in the lifetime of  fly $i$.
The 534 flies are
classified into long-lived, i.e. those who lived 44 days or longer,
and short-lived, i.e. those who died before the end of the 43rd day
after birth. In the data set, there are 256 short-lived, and 278
long-lived flies.
This classification naturally defines two samples:
{\em Sample 1:}
 the egg-laying curves
$\{X_i(t) ($resp. $Y_i(t)),\  0< t \le 30, \ i=1,2, \ldots, 256\}$ of the  short-lived
flies.
{\em Sample 2:}
 the egg-laying curves
$\{X_j^*(t) ($resp. $Y_j^*(t)),\  0< t \le 30, \ j=1,2, \ldots, 278\}$ of the  long-lived
flies.
The egg-laying curves are very irregular; Figure \ref{f:egg-curves}
shows ten (smoothed) curves of short- and long-lived flies for 
version 1, Figure \ref{f:egg-curves-rel} shows ten (smoothed) 
curves for version 2 (both using a B-spline basis for the representation).

\begin{figure}
\caption{Normal QQ--plots for the scores of the version 2
medfly data with respect
to the first two Fourier basis functions. Left  -- sample 1,
Right -- sample 2.}
\begin{center}
\includegraphics[scale = 0.7]{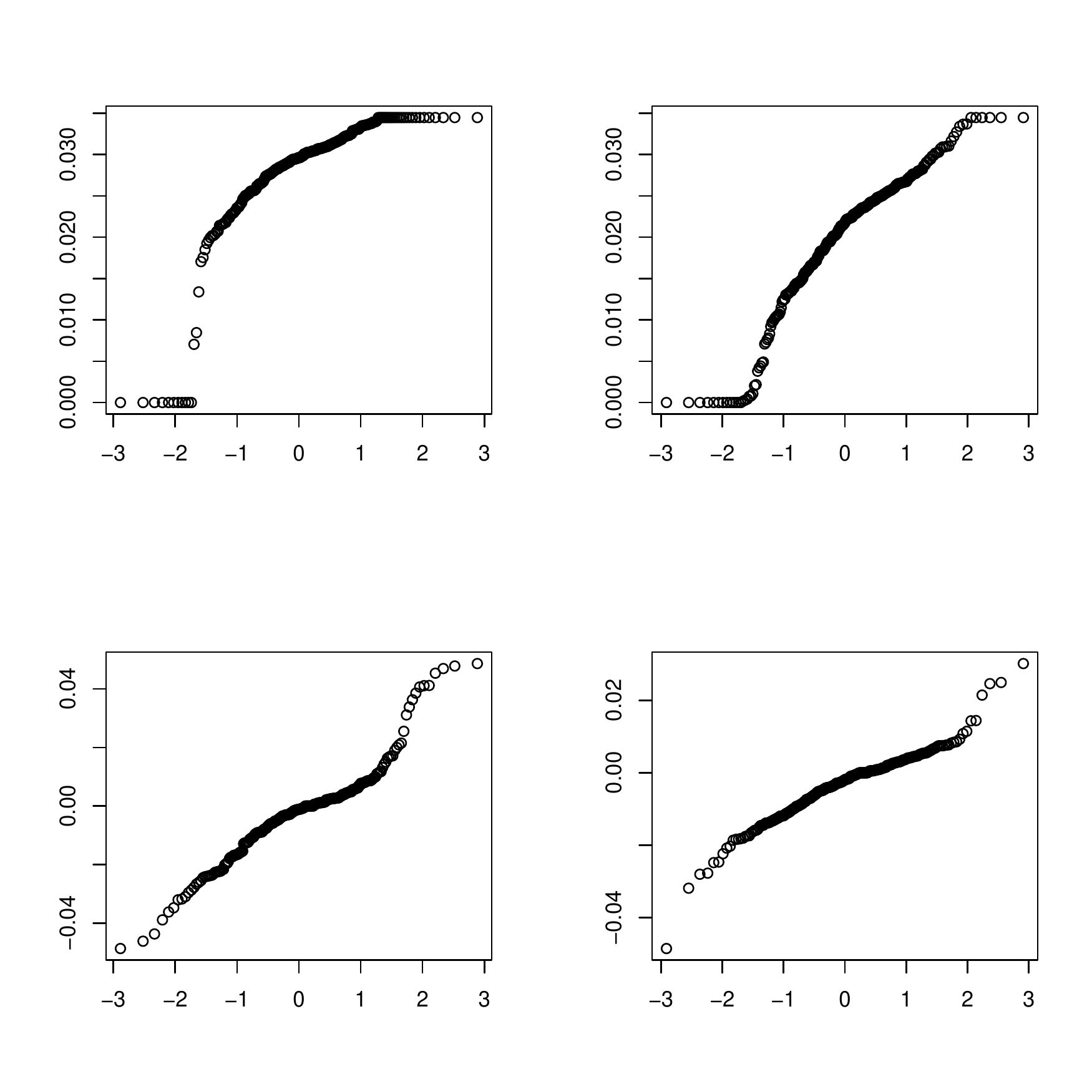}
\end{center}
\label{f:qq}
\end{figure}

\begin{table}
\caption{P--values (in percent) of the test based on
statistics $\what T_1$ and $\what T_2$
applied to absolute medfly data. Here $f_p$ denotes the fraction of the sample variance explained by the first $p$ FPCs, i.e. $f_p = (\sum_{k=1}^p \hat\lambda_k) / (\sum_{k = 1}^{N+M} \hat\lambda_k)$.}
\begin{center}
\begin{tabular}{|c|cccccccc|}
\hline
\multicolumn{9}{|c|}{P-values}\\
\hline
$p$      & 2 & 3  & 4 & 5 & 6 & 7 & 8 & 9 \\
\hline
$\what T_1$ & 82.70 & 36.22 & 30.59 & 63.84 & 37.71 & 39.03 & 33.77 & 34.77\\
$\what T_2$ & 0.54 & 0.13 & 0.11 & 0.12 & 0.02 & 0.00 & 0.00 & 0.00\\
\hline
$f_p$ & 72.93 & 78.36 & 81.87 & 83.94 & 85.62 & 87.08 & 88.49 & 89.72\\
\hline
\end{tabular}
\end{center}
\label{flies-absolute-bspline}
\end{table}

\begin{table}
\caption{P--values (in percent) of the test based on
statistics $\what T_1$ applied to relative medfly data;
 $f_p$ denotes the fraction of the sample variance explained by the 
 first $p$ FPCs, i.e. 
 $f_p = (\sum_{k=1}^p \hat\lambda_k) / (\sum_{k = 1}^{N+M} \hat\lambda_k)$.}
\begin{center}
\begin{tabular}{|c|ccccccc|}
\hline
\multicolumn{8}{|c|}{P-values}\\
\hline\hline
$p$   & 2 & 3  & 4 & 5 & 6 & 7 & 8 \\
\hline
$\what T_1$ & 0.14 & 0.06 & 0.33 & 1.50 & 3.79 & 4.53 & 10.28 \\
\hline
$f_p$ & 33.99 & 44.08 & 52.72 & 59.04 & 65.08 & 70.40 & 75.29\\
\hline\hline
$p$ & 9 & 10 & 11 & 12 & 13 & 14 & 15\\
\hline
$\what T_1$ & 5.51 & 2.78 & 5.32 & 3.21 & 1.78 & 6.28 & 3.80\\
\hline
$f_p$ & 79.91 & 83.72 & 86.58 & 89.02 & 91.34 & 93.30 & 95.03\\
\hline
\end{tabular}
\end{center}
\label{flies-relative-bspline}
\end{table}

Table~\ref{flies-absolute-bspline} shows the P--values for the  absolute egg-laying counts (version 1). For the statistic $\what T_1$ the null hypothesis cannot be rejected irrespective of the choice of $p$. For the statistic $\what T_2$, the result of the test varies depending on the choice of $p$. As explained in Section~\ref{sec1}, the usual recommendation
is  to use
the values of $p$ which explain 85 to 90 percent of the variance. 
For such values of $p$,  $\what T_2$ leads to a clear rejection. Since this test 
has however overinflated size, we conclude that there is little 
evidence that the covariance structures of version 1 curves for
long-- and short--lived flies are different. 
For the version 2 curves, the statistic $\what T_2$ yields P--values equal to zero (in machine precision), potentially indicating that the covariance structures 
for the short-- and long--lived flies are different. 
The assumption of a normal distribution is however questionable, as the QQ-plots in Figure~\ref{f:qq} show. These QQ-plots are constructed for the inner products $\langle Y_i, e_k \rangle$
and $\langle Y_i^*, e_k \rangle$, where the $Y_i$  are the curves from one of the samples (we cannot pool the data to construct QQ-plots because we test if the stochastic structures are different), and $e_k$  is the $k$th element of the Fourier basis. The normality of a functional sample implies the normality of all projections onto a complete orthonormal system 
. For $\langle X_i, e_k \rangle$, the QQ-plots show  a strong deviation from a straight line for some projections.
Almost all  projections $\langle Y_i, e_k \rangle$ have  QQ-plots indicating  a strong deviation 
from normality. It is therefore important to apply the robust test based on the statistic  $\what T_1$. The corresponding P--values for version 2  are displayed in Table~\ref{flies-relative-bspline}. For most values of $p$, these P--values
indicate the rejection of $H_0$. Many of them hover around 
the 5 percent level, but since the test is conservative, 
we can with confidence view them as favoring $H_A$. 

The above application confirms the properties of the statistics
established through the simulation study. It shows that while there 
is little evidence that the covariance structures for the 
absolute counts are different, there is strong evidence that they
are different for relative counts.

\section{Proofs of the results of Section \ref{sec1}}\label{sec2}

The proof of Theorem~\ref{thm 1.1} follows from several lemmas,
which we establish first.
We can and will assume without loss of generality that $\mu(t)=\mu^\ast (t)=0$ for all $t\in[0,1]$.

We will use the identity
\begin{align}\label{eq 2.1}
&\frac 1{N^{1/2}} \Sum^N_{k=1}\left(X_k(t)-\overline X_N (t)\right)
\left(X_k(s)-\overline X_N (s)\right)\\
&= \frac 1{N^{1/2}} \Sum^N_{k=1} X_k (t) X_k(s)- N^{1/2}
\overline X_N (t) \overline X_N(s),\notag
\end{align}
and an analogous identity for the second sample.

Our first lemma establishes bounds in probability which will
often be used in the proofs.

\begin{lem}\label{lem 2.1}
Under the assumptions of  Theorem~\ref{thm 1.1}, as $N,M\to \infty$,
\begin{align}
&\label{eq 2.2} \left\Vert N^{-1/2} \Sum^N_{k=1} \left\{ X_k(t) X_k (s)
-C(t,s)\right\}\right\Vert=O_P(1),
\\ \label{eq 2.3}
&\left\Vert N^{1/2} \overline X_N(t)\right\Vert = O_P(1),
\end{align}
and
\begin{align}
&\label{eq 2.4}\left\Vert M^{-1/2} \Sum^M_{k=1} \left\{ X^\ast_k(t) X^\ast_k (s)
-C^\ast(t,s)\right\}\right\Vert= O_P(1),
\\
&\label{eq 2.5}\left\Vert M^{1/2} \overline X^\ast_M(t)\right\Vert = O_P(1).
\end{align}
\end{lem}

\begin{proof}[\bf Proof.]
First we note that
\begin{align*}
&E\int^1_0 \int^1_0\left [\frac 1{N^{\frac 12}}
\Sum^N_{k=1}\left\{ X_k(t) X_k(s)- C(t,s)\right\}
\right ]^2 dt\,ds
= \int^1_0 \int^1_0 E\left\{ X_1 (t) X_1(s)-C (t,s)\right\}^2 dt\,ds,
\end{align*}
so, by Markov's inequality, we have
\[
\left\Vert \frac 1{N^{\frac12}}
\Sum^N_{k=1} \left\{X_k (t) X_k(s)-C(t,s)\right\} \right\Vert^2=O_P (1).
\]
Similar arguments yield (\ref{eq 2.3}) -- (\ref{eq 2.5}).
\end{proof}

The next lemma shows that the estimation of the
mean functions, cf.  the definition of
the projections $\what a_k(i)$ and  $\what a_k^*(j)$ in \eqref{eq 3.2} and \eqref{eq 3.3}, has
an asymptotically negligible effect.

\begin{lem}\label{lem 2.2}
Under the assumptions of  Theorem~\ref{thm 1.1},
for all $1\le i,j\le p$, as $N,M\to \infty$,
\[
N^{1/2} \what\Delta_N (i,j)=\frac 1{N^{1/2}}
\Sum^N_{k=1} \langle X_k, \what\vph_i\rangle
\langle X_k, \what\vph_j\rangle + O_P \left( N^{-1/2}\right)
\]
and
\[
 M^{1/2}\what\Delta^\ast_M (i,j) = \frac 1{M^{1/2}} \Sum^M_{k=1}
 \langle X^\ast_k, \what\vph_i \rangle
\langle X^\ast_k,\what\vph_j \rangle + O_P
 \left( M^{-1/2}\right).
\]
\end{lem}

\begin{proof}[\bf Proof.]
Using (\ref{eq 2.1}) and (\ref{eq 2.3}) we have by the Cauchy-Schwarz inequality,
\begin{align*}
&\left|\int^1_0 \int^1_0 N^{1/2} \overline X_N (t) \overline X_N(s)
\what\vph_i (t) \what\vph_j (s) dt\,ds\right|
\\
&= N^{-1/2}\left\vert \int^1_0 N^{1/2} \overline X_N(t) \what\vph_i(t) dt \right\vert
\left\vert \int^1_0 N^{1/2} \overline X_N(s) \what\vph_j(s) ds \right\vert
\\
&\le N^{-1/2}\left( \int^1_0 \left(N^{1/2} \overline X_N(t)\right)^2 dt
 \int^1_0 \what\vph^2_i(t) dt \right)^{1/2}
 \left(\int^1_0   \left(N^{1/2} \overline X_N(s)\right)^2 ds
 \int^1_0 \what\vph^2_j(s) ds \right)^{1/2}
 \\
 &= N^{-1/2} \int^1_0 \left(N^{1/2} \overline X_N(t)\right)^2 dt
 \\
 &= O_P \left( N^{-1/2}\right).
\end{align*}
The second part can be proven in the same way.
\end{proof}

We now state  bounds on the distances between the estimated and 
the population eigenvalues and eigenfunctions. These bounds are true
under the null hypothesis, and extend the corresponding one sample
bounds.

\begin{lem}\label{lem 2.3}
If the conditions of Theorem \ref{thm 1.1} are satisfied, then, as $N,M\to \infty$,
\[
\Max_{1\le i\le p} |\what \lam_i -\lam_i| = O_P \left((N+M)^{-1/2}\right)
\]
and
\[
\Max_{1\le i\le p} \|\what\vph_i -\what c_i\vph_i\|= O_P
\left((N+M)^{-1/2}\right),
\]
where
\[
\what c_i=\what c_i (N,M)= {\rm sign} (\langle \what\vph_i, \vph_i\rangle).
\]
\end{lem}

\begin{proof}[\bf Proof.]
It follows from (\ref{eq 2.2}) -- (\ref{eq 2.5}) and the assumption
$C=C^\ast$ that
\[
\left\Vert \what R_{N,M} - C \right\Vert =
O_P \left(\left( N^{1/2} + M^{1/2}\right)\Big/(N+M)\right),
 \]
and since $N^{1/2} + M^{1/2} \le 2(N+M)^{1/2}$,
the result follows from the corresponding one sample bounds,
see e.g. Chapter 2 of \citet{HKbook}.
\end{proof}

Lemma \ref{lem 2.3} now allows us to replace the
estimated eigenfunctions by their population counterparts.
The random signs $\what c_i$ must appear in the formulation
of Lemma~\ref{lem 2.4}, but they cancel in the subsequent results.

\begin{lem}\label{lem 2.4}
If the conditions of Theorem \ref{thm 1.1} are satisfied, then, for all $1\leq i,j \leq p$, as $N,M\to \infty$,
\begin{align*}
&\left(\frac{NM}{N+M}\right)^{1/2}
\left(\what\Delta_N (i,j) -\what\Delta^\ast_M (i,j)\right)
\\
&=\left(\frac{NM}{N+M}\right)^{1/2} \left\{\frac 1N \Sum^N_{k=1}
\langle X_k,\what c_i \vph_i\rangle \langle X_k,\what c_j \vph_j\rangle
-\frac 1M \Sum^M_{k=1} \langle X^\ast_k,\what c_i \vph_i\rangle
\langle X^\ast_k,\what c_j \vph_j\rangle  \right\}
+ o_P (1).
\end{align*}
\end{lem}
\begin{proof}[\bf Proof.]
We write
\begin{align*}
&\frac 1N \Sum^N_{k=1} \langle X_k,\what\vph_i\rangle \langle X_k,\what\vph_j\rangle
- \int^1_0 \int^1_0 C(t,s)\what\vph_i(t) \what\vph_j(s) dt~ds
\\
&= N^{1/2} \int^1_0 \int^1_0 \left\{\frac 1{N^{1/2}} \Sum^N_{k=1}
\big(X_k (t) X_k (s)-C (t,s)\big)\right\} \what\vph_i (t)\what\vph_j (s) dt\,ds.
\end{align*}
Using  Lemmas \ref{lem 2.1} -- \ref{lem 2.3} we get
\begin{align*}
&\left\vert \int^1_0 \int^1_0 \left\{ \frac 1{N^{1/2}} \sum^N_{k=1}
\big(X_k (t) X_k (s)
-C (t,s)\big)\right\} \big(\what\vph_i(t) \what\vph_j (s)
-\what c_i \vph_i(t)\what c_j \vph_j (s)\big)dt \,ds\right\vert
\\
&=
\Bigg\vert   \int^1_0 \int^1_0 \left\{ \frac 1{N^{1/2}} \sum^N_{k=1}
\big(X_k (t) X_k (s)-C(t,s)\big)\right\}
\\
&\qquad\times \left\{\big(\what\vph_i (t) - \what c_i \vph_i(t)\big)\what\vph_j (s) + \what c_i \vph_i (t)
\big(\what\vph_j (s)-\what c_j \vph_j (s)\big) \right\}dt\,ds\Bigg\vert
\\
&\le \Bigg( \int^1_0 \int^1_0 \left\{ \frac 1{N^{1/2}} \sum^N_{k=1}
\big(X_k (t) X_k (s)-C (t,s)\big)\right\}^2 dt\,ds
\\
&\qquad\times \int^1_0 \int^1_0 \big(\what\vph_i(t)-\what c_i \vph_i(t)\big)^2
\what\vph^2_j (s) dt\,ds\Bigg)^{1/2}
\\
&\quad +\Bigg(\int^1_0 \int^1_0 \left\{ \frac 1{N^{1/2}} \sum^N_{k=1}
\big(X_k (t) X_k (s)-C(t,s)\big)\right\}^2 dt\,ds
\\
&\qquad \times \int^1_0 \int^1_0 \vph^2_i (t) \big(\what\vph_j(s)-\what c_j
\vph_j(s)\big)^2  dt\,ds\Bigg)^{1/2}
\\
&= \left\Vert \frac 1{N^{1/2}} \Sum^N_{k=1} \big(X_k (t) X_k (s)-C(t,s)\big) \right\Vert
\big\{ \left\Vert \what\vph_i-\what c_i\vph_i\big\Vert
+\big\Vert \what\vph_j-\what c_j\vph_j\right\Vert\big\}
\\
&= o_P (1).
\end{align*}
Similar arguments give that
\[
\left\vert \int^1_0 \int^1_0 \left\{ \frac 1{M^{1/2}} \sum^M_{k=1}
\big(X^\ast_k (t)X^\ast_k(s) -C^\ast (t,s)\big)\right\}
\big\{\what\vph_i (t) \what\vph_j (s)
- \what c_i \vph_i (t) \what c_j \vph_j (s)\big\} dt\,ds \right\vert
= o_P (1).
\]
Since $C=C^\ast$, the lemma is proven.
\end{proof}

The previous lemmas isolated the main terms in
the differences $\what\Delta_N (i,j) -\what\Delta^\ast_M (i,j)$.
The following lemma describes the limits of these main terms
(without the random signs).

\begin{lem}\label{lem 2.5}
If the conditions of Theorem \ref{thm 1.1} are satisfied, then, as $N,M\to\infty$,
\[
\left\{ \Delta_{N,M} (i,j), 1\le i,j\le p\right\}~\tocalD~
\left\{ \Delta (i,j), 1\le i,j\le p\right\},
\]
where
\[
\Delta_{N,M} (i,j) =\left(\frac{NM}{N+M}\right)^{1/2}
\left\{ \frac 1N \Sum^N_{k=1} \langle X_k ,\vph_i\rangle
\langle X_k, \vph_j\rangle
- \frac 1M \Sum^M_{k=1} \langle X^\ast_k,\vph_i \rangle
\langle X^\ast_k,\vph_j \rangle\right\},
\]
and  $\left\{\Delta (i,j), 1\le i,j\le p\right\}$ is a Gaussian matrix with
$E\Delta (i,j)=0$ and
\begin{align*}
E\Delta (i,j)\Delta (i',j')&= (1-\Theta) \left\{ E\big(\langle X_1,\vph_i \rangle
\langle X_1,\vph_j\rangle\langle X_1,\vph_{i'} \rangle \langle X_1,\vph_{j'}\rangle\big)\right.
\\
&\qquad -\left. E\big(\langle X_1,\vph_i \rangle \langle X_1,\vph_j\rangle\big)
E\big(\langle X_1,\vph_{i'} \rangle \langle X_1,\vph_{j'}\rangle\big)\right\}
\\
&\quad + \Theta \left\{E\big(\langle X^\ast_1,\vph_i \rangle \langle X^\ast_1,\vph_j\rangle
\langle X^\ast_1,\vph_{i'} \rangle \langle X^\ast_1,\vph_{j'}\rangle\big) \right.
\\
&\qquad - \left. E\big(\langle X^\ast_1,\vph_i \rangle \langle X^\ast_1,\vph_j\rangle\big)
E\big(\langle X^\ast_1,\vph_{i'} \rangle \langle X^\ast_1,\vph_{j'}\rangle\big)\right\}.
\end{align*}
\end{lem}
\begin{proof}[\bf Proof.]
First we note that
\[
E\langle X_1,\vph_i\rangle\langle X_1,\vph_j\rangle=
E\langle X^\ast_1,\vph_{i'} \rangle \langle X^\ast_1,\vph_{j'}\rangle=
\left\{
\begin{array}{ccc}
0& \hbox{if}&i\not= j,
\\
\lam_i&\hbox{if}&i=j.
\end{array}
\right.
\]
Since $E\big(\langle X_1,\vph_i\rangle\langle X_1,\vph_j\rangle\big)^2 <\infty$
and $E\big(\langle X^\ast_1,\vph_i\rangle\rangle X^\ast_1,\vph_j\rangle\big)^2 <\infty,$
the multivariate central limit theorem
implies the result.
\end{proof}

Finally, we need an asymptotic approximation to the covariances
$\what L_{N,M} (k,k')$. Let
\begin{align*}
L_{N,M} (k,k')&=
(1-\Theta_{N,M})\left\{ \frac 1N
\Sum^N_{\ell=1}a_\ell (i) a_\ell(j) a_\ell(i') a_\ell(j')
-\big\langle \what{\mathfrak C}_N \what\vph_i,\what\vph_j\big\rangle \big\langle \what{\mathfrak C}_N \what\vph_{i'},\what\vph_{j'}\big\rangle\right\}
\\
&\qquad +\Theta_{N,M} \left\{ \frac 1M \Sum^M_{\ell=1}
a^\ast_\ell (i) a^\ast_\ell(j) a^\ast_\ell(i')a^\ast_\ell(j')
-\big\langle \what{\mathfrak C}^\ast_M \what\vph_i,\what\vph_j\big\rangle \big\langle \what{\mathfrak C}^\ast_M \what\vph_{i'},\what\vph_{j'}\big\rangle \right\},
\end{align*}
where
\[
a_\ell(i)= \langle X_\ell, \vph_i\rangle \quad\hbox{and}\quad
a^\ast_\ell(i)= \langle X^\ast_\ell, \vph_i\rangle,
\]
and $i,j,i',j'$ are determined from $k$ and $k'$ as in
\eqref{index1} and \eqref{index2}.
\begin{lem}\label{lem 2.6}
If the conditions of Theorem \ref{thm 1.1} are satisfied, then for all
$1\le k,k'\le p(p+1)/2,$
\[
\what L_{N,M} (k,k')-\what c_i \what c_j \what c_{i'} \what c_{j'}
L_{N,M} (k,k')= o_P (1)
\quad\text{as}\quad N,M\to\infty,
\]
where $(i,j)$ and $(i',j')$ are determined from $k$ and $k'$ as in
\eqref{index1} and \eqref{index2}.
\end{lem}
\begin{proof}[\bf Proof.]
The result follows from Lemma \ref{lem 2.3} along the lines of the proof of
Lemma \ref{lem 2.4}
\end{proof}
\begin{proof}[\bf Proof of Theorem \ref{thm 1.1}.]
According to Lemma \ref{lem 2.2} and Lemmas \ref{lem 2.4} -- \ref{lem 2.6},
the asymptotic
distribution of $\what{\bxi}^T_{N.M} \what L^{-1}_{N,M}
\what{\bxi}_{N,M}$
does not depend on the signs $\what c_1,\ldots,\what c_p$,
so it is sufficient
to prove the result for $\what c_1=\ldots =\what c_p=1$.
The law of large numbers
yields that
\begin{align}\label{eq 5.5}
L_{N,M}(k,k')\mathop{\longrightarrow}\limits^P L(k,k'),
\end{align}
where
\begin{align}
L(k,k')&= (1-\Theta) \left\{ E \big(a_1(i) a_1(j) a_1(i') a_1(j')\big)-
E\big( a_1(i) a_1(j) a_1(i') a_1(j')\big)\right\}\label{eq 5.8}
\\
&\qquad + \Theta \left\{E \big(a^\ast_1(i) a^\ast_1(j) a^\ast_1(i') a^\ast_1(j')\big)-
E\big(a^\ast_1(i) a^\ast_1(j) a^\ast_1(i') a^\ast_1(j')\big)\right\}.\notag
\end{align}
Now the result follows from Lemmas \ref{lem 2.2}, \ref{lem 2.4} and
\ref{lem 2.5}
\end{proof}

\begin{proof}[\bf Proof of Theorem~\ref{thm 3.1}.]
We continue to assume  that $\mu=\mu^\ast=0.$
This means that, under $H_0,$ $X_1,\ldots, X_N, X^\ast_1,\ldots, X^\ast_M$
are independent and identically distributed Gaussian processes.
Hence $\langle X_1,\vph_i\rangle =\lam^{1/2}_i N_i,$
$\langle X_1,\vph_j\rangle=\lam^{1/2}_j N^\ast_j,$ where
$N_i, N^\ast_j, 1\le i,j\le p$ are independent standard normal random
variables. We have already pointed out that
\[
E\langle X_1 ,\vph_i\rangle \langle X_1,\vph_j\rangle=
\left\{\begin{array}{ccc}
0& \hbox{if}& i\not= j,
\\
\lam_i& \hbox{if}& i=j,
\end{array}
\right.
\]
and
\[
E\langle X_1 ,\vph_{i'}\rangle \langle X_1,\vph_{j'}\rangle=
\left\{\begin{array}{ccc}
0& \hbox{if}& i'\not= j',
\\
\lam_{i'}& \hbox{if}& i=j'.
\end{array}
\right.
\]
If $i=j=i'=j',$ then
\begin{align*}
&E \big(\langle X_1,\vph_i\rangle\big)^4 -
\big(E (\langle X_1,\vph_i\rangle)^2\big)^2
= \lam^2_i E\big( N^4_i - (EN^2_i)^2\big)
= 2\lam^2_i .
\end{align*}
If $i=i'$ and $j=j'$ $(i\not=j),$ then
\[
E \langle X_1,\vph_i\rangle \langle X_1,\vph_j\rangle
\langle X_1,\vph_{i'}\rangle \langle X_1,\vph_{j'}\rangle = \lam_i \lam_j.
\]
In all other cases
\[
E \langle X_1,\vph_i\rangle \langle X_1,\vph_j\rangle
\langle X_1,\vph_{i'}\rangle \langle X_1,\vph_{j'}\rangle = 0.
\]

Hence $\Delta (i,j),$ $1\le i\le j\le p,$ are independent
normal random variables with mean $0$ and
\[
E\Delta^2 (i,j)=
\left\{\begin{array}{ccc}
\lam_i \lam_j& \hbox{if}& i\not= j,
\\
2\lam^2_i& \hbox{if}& i=j.
\end{array}
\right.
\]
Now the result follows from Lemmas \ref{lem 2.1} -- \ref{lem 2.5}.
\end{proof}

\begin{proof}[\bf Proof of Theorem~\ref{thm 3.3}.]
 First we observe that by the law of large numbers we have
 \[
  \Int_0^1\Int_0^1(\what R_{N,M}(t,s) - R(t,s))^2 dt\,ds = o_P(1).
 \]
Hence using the result in section VI.1. of \citet{gohberg:goldberg:kaashoek:1990} (cf. Lemmas 2.2 and 2.3 in \citet{HKbook}) we get that
\begin{align}\label{eq 5.6}
 \Max_{1\leq i\leq p}\left|\what\lambda_i - \lambda_i\right| = o_P(1)
\end{align}
and
\begin{align}\label{eq 5.7}
 \Max_{1\leq i\leq p}\Vert\what\vph_i - \what c_i\vph_i\Vert = o_P(1),
\end{align}
where $\what c_i = \what c_i(N,M) = \sign(\langle \what \vph_i,\vph_i\rangle).$ Relations \eqref{eq 5.6} and \eqref{eq 5.7} show that Lemma \ref{lem 2.3} remains true. It follows from the law of large numbers and \eqref{eq 5.7} that for all $1\leq i,j\leq p$
\begin{align*}
 &\left| \what\Delta_N (i,j) - \what\Delta_M^\ast (i,j) - \what c_i\what c_j\Int_0^1\Int_0^1\left(C(t,s)-C^\ast(t,s)\right)\vph_i(t)\vph_j(s)dt\,ds\right|\\
 =& \left|\Int_0^1\Int_0^1\left(\what C_N(t,s)-\what C_M^\ast(t,s)\right)\what\vph_i(t)\what\vph_j(s)dt\,ds
 - \what c_i\what c_j\Int_0^1\Int_0^1\left(C(t,s)-C^\ast(t,s)\right)\vph_i(t)\vph_j(s)dt\,ds\right|\\
 \leq&\left|\Int_0^1\Int_0^1\left(\what C_N(t,s)-C(t,s)-\left(\what C_M^\ast(t,s) - C^\ast(t,s)\right)\right)\what\vph_i(t)\what\vph_j(s)dt\,ds\right|\\
&+\left|\Int_0^1\Int_0^1\left(C(t,s)-C^\ast(t,s)\right)\left(\what\vph_i(t)\what\vph_j(s) - \what c_i \vph_i(t)\what c_i \vph_j(s)\right)dt\,ds\right|\\
\leq & \left\|\what C_N - C\right\Vert + \left\Vert\what C_M^* - C^\ast\right\Vert + \left\|C-C^\ast\right\|\left\|\what\vph_i\what\vph_j - \what c_i \vph_i\what c_i \vph_j\right\|\\
=& ~o_P(1),
\end{align*}
where the fact that $\Vert\vph_i\Vert = 1 = \Vert\what\vph_i\Vert$ was used.
Hence the proof of \eqref{eq 3.5} is complete. It is also clear that \eqref{eq 3.5} implies \eqref{eq 3.6}.\\
Next we observe that Lemma \ref{lem 2.6} and  \eqref{eq 5.5} remain true under the alternative. Now by some lengthy calculations it can be verified that $L$ given in \eqref{eq 5.8} is positive definite so that \eqref{eq 3.7} follows from \eqref{eq 3.6}.
\end{proof}

\bibliographystyle{plainnat}

\begin{thebibliography}{10}
\bibitem[Abadir \& Magnus(2005)]{abadir:magnus:2005}
Abadir, K. M. \& Magnus, J.R. (2005).
\newblock {\em Matrix algebra}.
\newblock Cambridge University Press, New York.

\bibitem[Benko et al.(2009)]{benko:hardle:kneip:2009}
Benko, M., H{\"a}rdle, W. \& A.~Kneip (2009).
\newblock Common functional principal components.
\newblock {\em Ann. Statist.}, {\bf 37}, 1-34.

\bibitem[Bosq(2000)]{bosq:2000}
Bosq, D. (2000).
\newblock {\em Linear processes in function spaces}.
\newblock Springer, New York.

\bibitem[Ferraty \& Romain(2011)]{ferraty:romain:2011}
Ferraty, F. \& Romain, Y., editors, (2011).
\newblock {\em The Oxford handbook of functional data analysis}.
\newblock Oxford University Press.

\bibitem[Ferraty \& Vieu(2006)]{ferraty:vieu:2006}
Ferraty, F. \& Vieu, P. (2006).
\newblock {\em {N}onparametric functional data analysis: Theory and
  practice}.
\newblock Springer, New York.

\bibitem[Gabrys et al.(2010)]{gabrys:horvath:kokoszka:2010}
Gabrys, R., Horv{\'a}th, L. \& Kokoszka, P. (2010).
\newblock Tests for error correlation in the functional linear model.
\newblock {\em J. Amer. Statist. Assoc.},
  {\bf 105}, 1113-1125.

\bibitem[Gervini(2008)]{gervini:2008}
Gervini, D. (2008).
\newblock Robust functional estimation using the spatial median and spherical
  principal components.
\newblock {\em Biometrika}, {\bf 95}, 587-600.

\bibitem[Gohberg et al.(1990)]{gohberg:goldberg:kaashoek:1990}
Gohberg, I., Goldberg, S. \& Kaashoek, M.A. (1990).
\newblock {\em Classes of linear operators.}
\newblock Operator {T}heory: {A}dvances and {A}pplications, {\bf 49},
\newblock Birkh\"auser, Basel.

\bibitem[Horv{\'a}th \& Kokoszka(2011+)]{HKbook}
Horv{\'a}th, L. \& Kokoszka, P. (2011+).
\newblock {\em Inference for functional data with applications}.
\newblock Springer Series in Statistics. Springer, New York.
\newblock Forthcoming.

\bibitem[Horv{\'a}th et al.(2011)]{horvath:kokoszka:reeder:2011}
Horv{\'a}th, L., Kokoszka, P. \& Reeder, R. (2011).
\newblock Estimation of the mean of functional time series and a two sample
  problem.
\newblock Technical report, University of Utah.

\bibitem[Horv{\'a}th et al.(2009)]{horvath:kokoszka:reimherr:2009}
Horv{\'a}th, L., Kokoszka, P. \& Reimherr, M. (2009).
\newblock Two sample inference in functional linear models.
\newblock {\em Canad. J. Statist.}, {\bf 37}, 571-591.

\bibitem[M{\"u}ller \& Stadtm{\"u}ller(2005)]{muller:stadtmuller:2005}
M{\"u}ller, H-G. \& Stadtm{\"u}ller, U. (2005).
\newblock Generalized functional linear models.
\newblock {\em Ann. Statist.}, {\bf 33}, 774-805.

\bibitem[Panaretos et al.(2010)]{panaretos:2010}
Panaretos, V.~M., Kraus, D. \& Maddocks, J.~H. (2010).
\newblock Second-order comparison of {G}aussian random functions and the
  geometry of {DNA} minicircles.
\newblock {\em J. Amer. Statist. Assoc.}, {\bf 105}, 670--682.

\bibitem[Ramsay et al.(2009)]{ramsay:hooker:graves:2009}
Ramsay, J., Hooker, G. \& Graves, S. (2009).
\newblock {\em Functional data analysis with {R} and {MATLAB}}.
\newblock Springer, New York.

\bibitem[Ramsay \& Silverman(2005)]{ramsay:silverman:2005}
Ramsay, J.~O. \& Silverman, B.~W. (2005).
\newblock {\em Functional data analysis}.
\newblock Springer, New York.

\bibitem[Reiss \& Ogden(2007)]{reiss:ogden:2007}
Reiss, P.~T. \& Ogden, R.~T. (2007).
\newblock Functional principal component regression and functional partial
  least squares.
\newblock {\em J. Amer. Statist. Assoc.}, {\bf 102}, 984-996.

\bibitem[Yao \& M{\"u}ller(2010)]{yao:muller:2010}
Yao, F. \& M{\" u}ller, H-G. (2010).
\newblock Functional quadratic regression.
\newblock {\em Biometrika}, {\bf 97}, 49-64.

\end{thebibliography}

\end{document}